\documentclass[10pt,a4paper,oneside]{article}
\usepackage[utf8]{inputenc}
\usepackage[english]{babel}
\pdfoutput=1

\usepackage{amsmath,amsfonts,amssymb,amsopn,amscd,amsthm}
\allowdisplaybreaks

\usepackage[yyyymmdd,hhmmss]{datetime}

\usepackage{fancyhdr}
\newtheorem{Thm}{Theorem}

\newtheorem{Prop}[Thm]{Proposition}

\newtheorem{Alg}[Thm]{Algorithm}

\newcommand{\Then}{\quad\Rightarrow\quad}
\newcommand{\Iff}{\quad\Leftrightarrow\quad}

\newcommand{\Ordinals}{\mathbb{N}}
\newcommand{\Cardinals}{\mathbb{N}_0}
\newcommand{\Integers}{\mathbb{Z}}
\newcommand{\Reals}{\mathbb{R}}
\newcommand{\Complex}{\mathbb{C}}

\newcommand{\Set}[1]{\{{#1}\}}

\newcommand{\Cardinality}[1]{|{#1}|}
\newcommand{\IntegerSet}[1]{[{#1}]}
\newcommand{\Iverson}[1]{[{#1}]}
\newcommand{\Modulus}[1]{|{#1}|}
\newcommand{\BigModulus}[1]{\left|{#1}\right|}
\DeclareMathOperator{\argmin}{argmin}
\DeclareMathOperator{\diam}{diam}
\newcommand{\MinOp}{\lor}
\newcommand{\MaxOp}{\land}

\newcommand{\Floor}[1]{\lfloor{#1}\rfloor}
\newcommand{\DiffVecOp}{\triangledown}

\newcommand{\Space}{X}
\newcommand{\Metric}{d}
\newcommand{\FinVol}{\Lambda}

\newcommand{\IsRHC}[2]{\mathcal{H}_{#1}(#2)}
\newcommand{\GenFunSym}{Z}
\newcommand{\GenFunAt}[2]{\GenFunSym(#1,#2)}
\newcommand{\SingG}{{\rho^\star}}

\newcommand{\CaseK}{k}
\newcommand{\CaseC}{c}

\newcommand{\BlockKG}[2]{z_\CaseK(#1,#2)}

\newcommand{\BlockCG}[2]{z_\CaseC(#1,#2)}
\newcommand{\BlockCR}[1]{z_\CaseC(#1)}

\newcommand{\CondBlockKG}[2]{a_\CaseK(#1,#2)}
\newcommand{\CondBlockKR}[1]{a_\CaseK(#1)}
\newcommand{\CondBlockCG}[3]{a_\CaseC(#1,#2,#3)}
\newcommand{\CondBlockCR}[2]{a_\CaseC(#1,#2)}

\newcommand{\RootK}[1]{\rho_\CaseK(#1)}
\newcommand{\RootC}[1]{\rho_\CaseC(#1)}

\newcommand{\SingK}{{\rho^\star_\CaseK}}
\newcommand{\SingC}{{\rho^\star_\CaseC}}

\newcommand{\ContractionK}{f_k}
\newcommand{\ContractionC}{f_c}

\newcommand{\ValK}{\frac{k^k}{(k+1)^{k+1}}}
\newcommand{\ValC}{\frac{1}{e}}

\newcommand{\FreeK}[1]{F_\CaseK(#1)}
\newcommand{\FreeC}[1]{F_\CaseC(#1)}

\newcommand{\InvK}[1]{{\lambda_\CaseK(#1)}}
\newcommand{\InvC}[1]{{\lambda_\CaseC(#1)}}

\newcommand{\NuX}[1]{{\nu^{#1}}}
\newcommand{\TMO}{^{-1}}

\newcommand{\Cluster}[1]{G(#1)}
\newcommand{\Ursell}[1]{U(#1)}

\newcommand{\SpanningTreesOf}[1]{\mathcal{T}_{#1}}
\newcommand{\SpanningGraphsOf}[1]{\mathcal{G}_{#1}}
\newcommand{\PartitionScheme}{S}
\newcommand{\SingletonTrees}[1]{\mathcal{S}_{#1}}

\newcommand{\SchemeOne}{One}
\newcommand{\SchemePenrose}{Pen}

\newcommand{\Edges}[1]{E(#1)}
\newcommand{\Tree}{\mathbb{T}}
\newcommand{\SpanningTrees}[1]{\mathcal{T}_{#1}}
\newcommand{\TreePath}[2]{{P(#1,#2)}}
\newcommand{\Parent}[1]{{\mathfrak{p}(#1)}}

\newcommand{\Progressing}{\mathbf{P}}
\newcommand{\Reversing}{\mathbf{R}}
\newcommand{\SideStatus}{d}
\newcommand{\SideClass}[2]{[#1]_{(#2)}}
\newcommand{\SideEquivalent}[1]{\sim_{(#1)}}
\newcommand{\SideEquivalentKMinusOne}[1]{\SideEquivalent{k-1}}
\newcommand{\SideEquivalentK}{\SideEquivalent{k}}
\newcommand{\SideEquivalentKPlusOne}{\SideEquivalent{k+1}}

\newcommand{\ActionRemove}{REMOVE}
\newcommand{\ActionSelect}{SELECT}

\newcommand{\StepStyle}[1]{\textbf{(#1)}}
\newcommand{\StepDirProgBoundary}{\StepStyle{dpb}}
\newcommand{\StepDirProgIgnored}{\StepStyle{dpi}}
\newcommand{\StepDirProgParent}{\StepStyle{dpp}}
\newcommand{\StepDirProgUncles}{\StepStyle{dpu}}
\newcommand{\StepDirProgCousins}{\StepStyle{dpc}}

\newcommand{\StepDirRevBoundary}{\StepStyle{drb}}
\newcommand{\StepDirRevIgnored}{\StepStyle{dri}}
\newcommand{\StepDirRevParent}{\StepStyle{drp}}
\newcommand{\StepDirRevUnclesReversing}{\StepStyle{drur}}

\newcommand{\StepDirRevCousins}{\StepStyle{drc}}

\newcommand{\AdmissibleEdgesTree}[1]{\mathcal{A}_{#1}(\Tree)}

\newcommand{\ConflictingEdgesTree}[1]{\mathcal{C}_{#1}(\Tree)}
\newcommand{\Abstract}{
\begin{abstract}
We revisit the smallest non-physical singularity of the hard-sphere model in one dimension, also known as Tonks gas. We give an explicit expression of the free energy and reduced correlations at negative real fugacity and elaborate the nature of the singularity: the free energy is right-continuous, but its derivative diverges. We derive these results in several novel ways: First, by scaling up the discrete solution. Second, by an inductive argument on the partition function à la Dobrushin. Third, by a perfect cluster expansion counting the Penrose trees in the Mayer expansion perfectly. Fourth, by an explicit construction of Shearer's point process, the unique $R$-dependent point process with an $R$-hard-core. The last connection yields explicit and optimal lower bounds on the avoidance function of $R$-dependent point processes on the real line.
\end{abstract}
}

\newcommand{\TitleFull}{Shearer's point process and the hard-sphere model in one dimension}

\newcommand{\TitleShort}{Shearer's PP 1D}
\newcommand{\TitlePDF}{Shearer's\ PP\ \&\ hard-sphere\ in\ 1D} 

\newcommand{\AuthorsFull}{Hofer-Temmel Christoph (math@temmel.me)\thanks{The author acknowledges the support of the VIDI project ``Phase transitions, Euclidean fields and random fractals'', NWO 639.032.916}}

\newcommand{\AuthorsShort}{Hofer-Temmel Christoph}

\newcommand{\Keywords}{Keywords:
Shearer's point process,
hard-sphere model,
Lovász Local Lemma,
Tonks gas,
cluster expansion,
Lambert W function}

\newcommand{\Head}{
 \maketitle
 \Abstract{}
 \Keywords{}\\\MSC{}
 \tableofcontents
 \listoffigures
}

\usepackage[pdftex%
 ,bookmarks=true%
 ,bookmarksopen=true%
 ,bookmarksopenlevel=3%
 ,colorlinks=false%
 ,pdftitle=Arxiv\ \pdfdate:\ \TitlePDF{}%
]{hyperref}

\title{\TitleFull{}}
\author{\AuthorsFull{}}
\date{}
\begin{document}

\Head{}

\section{Introduction}
\label{sec_intro}

The present paper investigates the behaviour of the partition function of the one-dimensional hard-sphere model~\cite{Ruelle__StatisticalMechanics_RigorousResults__Ben_1969}, also known as
Tonks gas~\cite{Tonks__TheCompleteEquationOfStateOfOneTwoAndThreeDimensionalGasesOfElasticHardSpheres__PhysRev_1936}, at small fugacities. Although the one-dimensional case is not as physically relevant as the higher-dimensional cases, it is interesting because explicit and beautiful calculations are possible. This paper collects known results and adds new ones to complete the picture.\\

The partition function has a smallest non-physical singularity at negative real fugacity. This fugacity limits the range where high-temperature (equals low fugacity) cluster expansions à la Mayer~\cite{Mayer_Montroll__MolecularDistribution__JCP_1941} are possible and where analyticity
of the free energy and reduced correlations is easily established. The value of the singularity is already known for a long time and has been derived by cluster expansion techniques~\cite{Brydges_Imbrie__DimensionalReductionFormulasForBranchedPolymerCorrelationFunctions__JSP_2003,Bernardi__SolutionToACombinatorialPuzzleArisingFromMayersTheoryOfClusterIntegrals__SemLotComb_2007}.\\

We give explicit expressions for the free energy and ratios of the partition function at negative real activities between zero and the smallest non-physical singularity. Usually, only bounds are given (and possible). Not surprising, the Lampert W function~\cite{Corless_Gonnet_Hare_Jeffrey_Knuth__OnTheLambertWFunction__AdvCompMath_1996,Bernardi__SolutionToACombinatorialPuzzleArisingFromMayersTheoryOfClusterIntegrals__SemLotComb_2007} and two discrete relatives make their appearance. We obtain these expressions by different approaches: scaling from the discrete case, cluster expansions or alternative viewpoints on the partition function via a combinatoric relationship to the avoidance function of a particular one-dependent point process. As a consequence, we shed more light on the nature of singularity: the free energy is right-continuous at the singularity, whereas its derivative diverges at this location.\\

The discrete one-dimensional hard-sphere model lives on $\Integers$ with a hard-core radius parametrised by $k$. We collect the results about the discrete case~\cite{Liggett_Schonmann_Stacey__DominationByProductMeasures__AoP_1997,Mathieu_Temmel__KIndependentPercolationOnTrees__SPA,Temmel__ShearersMeasureAndStochasticDominationOfProductMeasures__JTP_2014} and specialise results from~\cite{Temmel__SufficientConditionsForUniformBoundsInAbstractPolymerSystemsAndExplorativePartitionSchemes__JSP_2014}. We want to point out, that in this case a purely inductive approach à la Dobrushin~\cite{Dobrushin__PerturbationMethodsOfTheTheoryOfGibbsianFields__LNM_1996} yields the full picture and cluster expansion is not needed at all. The first way to approach the continuous model is to scale the resulting quantities towards a range $1$ hard-sphere model on $\Reals$. We recover the ubiquitous Lambert's W function. The scaling works for the critical values, but not for all related quantities.\\

The other approaches work directly with the continuous model. In~\cite{Temmel__ShearersPointProcessTheHardSphereGasAndAContinuumLovaszLocalLemma}, we generalised the inductive approach of Dobrushin~\cite{Dobrushin__PerturbationMethodsOfTheTheoryOfGibbsianFields__LNM_1996} to the continuous setting. A variant thereof, inspired by the calculations in~\cite{Liggett_Schonmann_Stacey__DominationByProductMeasures__AoP_1997}, identifies the smallest non-physical singularity directly and yields bounds on some of the related quantities.\\

For the continuous case, we extend the tree-operator approach for cluster expansions of discrete polymer models from~\cite{Temmel__SufficientConditionsForUniformBoundsInAbstractPolymerSystemsAndExplorativePartitionSchemes__JSP_2014} to the hard-sphere model. We exhibit a partition scheme of the cluster adapted to the one dimensional setting. The singleton trees of the partition scheme have an explicit structure, which allows us to derive cluster expansion expressions exactly, instead of just bounding them. This approach was partly motivated by the fact that~\cite{Fernandez_Procacci_Scoppola__TheAnalyticityRegionOfTheHardSphereGas_ImprovedBounds__JSP_2007}, the current best known cluster expansion approach using tree-operators for the hard-sphere model, does not attain the critical value in one dimension. We describe the free energy at negative real fugacity explicitly up to this singularity. At the fugacity approaches the singularity coming from zero, the free energy expression converges, but its derivative diverges.\\

Finally, we use the connection between the partition function of a $R$-hard-sphere model and the avoidance function of Shearer's point process~\cite{Scott_Sokal__TheRepulsiveLatticeGasTheIndependentSetPolynomialAndTheLovaszLocalLemma__JSP_2005,Temmel__ShearersPointProcessTheHardSphereGasAndAContinuumLovaszLocalLemma}.
Shearer's point process is the unique one-independent point process with one-hard-core configurations. An alternative proof of a lower bound of the singularity comes from an explicit construction of Shearer's point process. The construction is a one-sided variant of a Matérn-type deletion rule~\cite{Matern__SpatialVariation_StochasticModelsAndTheirApplicationToSomeProblemsInForestSurveysAndOtherSamplingInvestigations__MFSS_1960}, which seems to seems to depend on the chordal (i.e., tree-like) structure of $\Reals$. See~\cite{Lehner_Temmel__CliqueTreesOfInfiniteLocallyFiniteChordalGraphs} for the case of chordal graphs. As some of the results in the discrete case have been used to solve problems in  $k$-dependent percolation on trees~\cite{Mathieu_Temmel__KIndependentPercolationOnTrees__SPA}, we intend the explicit expressions of the continuous case as a tool to treating dependent continuous Boolean percolation problems.\\

A related area of research are the combinatorics of the virial expansion~\cite{Jansen__ClusterAndVirialExpansionsForTheMultiSpeciesTonksGas__Arxiv_1503_02338v1,Ramawadh_Tate__VirialExpansionBoundsThroughTreePartitionSchemes__Arxiv_1501_00509v1}. These are not a topic of discussion here.\\

The following subsections introduce the key terms in detail. Section~\ref{sec_results} summarises the results. The scaling approach is in section~\ref{sec_scaling}, the inductive approach in section~\ref{sec_ind}, the tree-operator approach for cluster expansions in section~\ref{sec_ce} and the constructions of Shearer's point process in theorems~\ref{thm_shearerK} and~\ref{thm_shearerC}.

\subsection{The generating function}
\label{sec_gf}

Let $(\Space,\Metric)$ be a Polish space. A set of points in $\Space$ is $R$-hard-core, if all pairs of points have mutual distance at least $R$. Denote by $\IsRHC{.}$ the indicator function of $R$-hard-core point sets. On a bounded volume $\FinVol\Subset\Space$ and for weight $z\in\Complex$, we regard the generating function
\begin{equation}\label{eq_gf}
 \GenFunAt{z}{\FinVol}:=
 \sum_{n=0}^\infty \frac{z^n}{n!}
 \int_{\FinVol^n} \IsRHC{R}{\Set{x_1,\dotsc,x_n}}
 \prod_{i=1}^n dx_i\,,
\end{equation}
with integration to some standard measure on $\Space$. Let $A$ be a Borel set with $\diam(A)<R$. There is a deletion-contraction identity
\begin{equation}\label{eq_fe}
 \GenFunAt{z}{\FinVol}
 = \GenFunAt{z}{\FinVol\setminus A}
 + z \int_A \GenFunAt{z}{\FinVol\setminus B(x,R)} dx\,.
\end{equation}
For $\FinVol\Subset\Space$ fixed, the partition function has its smallest root on the negative real axis. A key role is played by the quantity
\begin{equation}\label{eq_sing}
 \SingG:=\sup\Set{\rho\ge 0\mid
  \forall \FinVol\Subset\Space:
  \GenFunAt{-\rho}{\FinVol}> 0
 }\,.
\end{equation}
It is a singularity of the cluster expansion
\begin{equation}\label{eq_ce}
 \log\GenFunAt{-\rho}{\FinVol}
 = -\sum_{n=1}^\infty \frac{\rho^n}{n!}
   \int_{\FinVol^n} \Ursell{x_1,\dotsc,x_n} dx_1\dotsm dx_n\,,
\end{equation}
where the Ursell coefficients~\cite{Ursell__TheEvaluationOfGibbsPhaseIntegralForImperfectGases__MPCPS_1927} $\Ursell{x_1,\dotsc,x_n}\in\Cardinals$ are combinatorial expressions depending on the whether pairs of points in $\Set{x_1,\dotsc,x_n}$ have distance less than $R$.

\subsection{Shearer's point process}
\label{sec_shearer}

A point process $\xi$ on $(\Space,\Metric)$ is $R$-dependent, if, for every pair of subsets $\FinVol_1,\FinVol_2$ with $\Metric(\FinVol_1,\FinVol_2)\ge R$, the projections $\FinVol_1\xi$ and $\FinVol_2\xi$ are independent. For a given intensity, there is at most one $R$-dependent point process with an almost-sure $R$-hard-core. This is Shearer's point process~\cite{Shearer__OnAProblemOfSpencer__Comb_1985,Temmel__ShearersPointProcessTheHardSphereGasAndAContinuumLovaszLocalLemma}.\\

The avoidance function of Shearer's point process with intensity $\rho$ is $\GenFunSym$ evaluated at weight $-\rho$. Thus, Shearer's point process exists, iff it has intensity $\rho\le\SingG$. If Shearer's point process exists, it has the minimal avoidance function among all $R$-dependent point processes of the same intensity.

\subsection{The hard-sphere model}
\label{sec_hardsphere}

The hard-sphere model with radius $R$ is the unique Markov point process on $(\Space,\Metric)$ with interaction range $R$ and an almost-sure $R$-hard-core. For a bounded subset $\FinVol\Subset\Space$, \emph{fugacity} $\lambda\in[0,\infty[$ and empty boundary conditions, its partition function is $\GenFunAt{\lambda}{\FinVol}$.\\

The behaviour of $\log\GenFunAt{z}{\FinVol}$ and ratios of the partition function at different volumes in low fugacity (high-temperature) setting are a topic of longstanding interest~\cite{Mayer_Montroll__MolecularDistribution__JCP_1941,Ruelle__StatisticalMechanics_RigorousResults__Ben_1969}. Uniform bounds (in both volume and fugacity) in a small disc in $\Complex$ around the origin are possible because of the inequality~\cite[Theorem 2.10]{Fernandez_Procacci__ClusterExpansionForAbstractPolymerModels_NewBoundsFromAnOldApproach__CMP_2007}:
\begin{equation}
 \BigModulus{\log\GenFunAt{z}{\FinVol}}\le -\log\GenFunAt{\Modulus{z}}{\FinVol}\,.
\end{equation}
Knowledge about the location of $\SingG$ is crucial to expansions of $\log\GenFunAt{z}{\FinVol}$. Because $\SingG$ represents a negative fugacity, it is also called a non-physical singularity.

\subsection{Lambert W function}
\label{sec_lampert}

The Lambert $W$ function~\cite{Corless_Gonnet_Hare_Jeffrey_Knuth__OnTheLambertWFunction__AdvCompMath_1996} are the solutions of the equation
\begin{equation*}
 W(z)\exp(W(z))=z
\end{equation*}
in the complex plane. For $\rho\in[0,\ValC]$, there is a real branch $W_0$ of solutions in $[0,1]$ solving the equation $W(-\rho)\exp(W(-\rho))=-\rho$ uniquely. This solution plays a key role in the low fugacity case of the Tonks gas~\cite[Section 3]{Caillol__SomeApplicationsOfTheLambertWFunctionToClassicalStatisticalMechanics__JPA_2003} and appears also in our results. The reason is that the Lambert W function has a close relationship with the generating function $T$ of rooted trees on $n$ labelled points~\cite[Section 2]{Corless_Gonnet_Hare_Jeffrey_Knuth__OnTheLambertWFunction__AdvCompMath_1996}: $T(x)=-W(-x)$ and $T$ converges, iff $x\in[0,\ValC]$.

\section{Results}
\label{sec_results}

\subsection{The discrete case}
\label{sec_K}

For $k\in\Cardinals$, we regard the space $(\Integers,\Modulus{.})$ with hard-core radius $k+1$ and integration with respect to the counting measure. The subscript $k$ marks related quantities. Instead of point processes, we talk of a Bernoulli random field (short BRF), i.e., a collection of $\Set{0,1}$-valued rvs. This section summaries already known results for comparison with the continuous case in section~\ref{sec_C}.\\

We use the convention $0^0=1$. Consider the map
\begin{equation}\label{eq_funK}
 \ContractionK:
  \quad[0,1]\to\left[0,\ValK\right]
  \quad x\mapsto x(1-x)^k\,.
\end{equation}
For $\rho\in\left[0,\ValK\right]$, let $\InvK{\rho}$ be the unique pre-image of $\ContractionK$ in $\left[0,\frac{1}{(k+1)}\right]$.

\begin{Thm}\label{thm_shearerK}
For $\rho\in[0,\ValK]$, let $X:=(X_n)_{n\in\Integers}$ be a Bernoulli product field with parameter $\InvK{\rho}$. Define the BRF $Y:=(Y_n)_{n\in\Integers}$ by
\begin{equation}\label{eq_construction_discrete}
 Y_n:=X_n \prod_{i=1}^k (1-X_{n-i})\,.
\end{equation}
Then $Y$ is Shearer's BRF on $\Integers$, for marginal probability $\rho$ and distance $(k+1)$. Shearer's BRF does not exist for higher marginal probabilities than $\ValK$.
\end{Thm}

The construction is from~\cite{Mathieu_Temmel__KIndependentPercolationOnTrees__SPA}, inspired by~\cite{Liggett_Schonmann_Stacey__DominationByProductMeasures__AoP_1997} and may be seen as a special case of the general construction on chordal graphs~\cite{Lehner_Temmel__CliqueTreesOfInfiniteLocallyFiniteChordalGraphs}. In particular, the BRF $Y$ is a \emph{$(k+1)$-block factor} of the BRF $X$. The non-existence follows from~\eqref{eq_mainK_sing} in theorem~\ref{thm_mainK}.\\

We have $\BlockKG{\rho}{n}:=\GenFunAt{-\rho}{\Set{1,\dotsc,n}}$, with $\BlockKG{\rho}{0}=1$ and $\BlockKG{\rho}{1}=1-\rho$. Let $\RootK{n}$ be the smallest root of $\BlockKG{\rho}{n}$. If $\BlockKG{\rho}{n-1}>0$, then let $\CondBlockKG{\rho}{n}:=\frac{\BlockKG{\rho}{n}}{\BlockKG{\rho}{n-1}}$. Let $\displaystyle\FreeK{\rho}:=\lim_{n\to\infty}\frac{\log\BlockKG{\rho}{n}}{n}$. The singularity reduces to
\begin{equation*}
 \SingK
 :=
 \sup\Set{\rho\ge 0\mid
  \forall n\in\Ordinals: \CondBlockKG{\rho}{n}>0
 }
 =\inf\Set{\RootK{n}\mid n\in\Cardinals}\,.
\end{equation*}

\begin{Thm}\label{thm_mainK}
For each $k\in\Cardinals$, we have
\begin{subequations}\label{eq_mainK}
\begin{gather}
 \label{eq_mainK_sing}
  \SingK = \ValK\,,
 \\\label{eq_mainK_roots_mon}
  \forall n\in\Cardinals:
  \quad
  \RootK{n}>\RootK{n+1}\text{ with equality if $k=0$}\,,
 \\\label{eq_mainK_roots_lim}
  \lim_{n\to\infty} \RootK{n} = \SingK\,,
 \\\label{eq_mainK_cond_mon}
  \forall\rho\in]0,\SingK],n\in\Ordinals:
  \quad
  \CondBlockKG{\rho}{n} > \CondBlockKG{\rho}{n+1}
  \text{ with equality if $k=0$}
  \,,
 \\\label{eq_mainK_block_mon}
  \forall\rho\in]0,\SingK],n\in\Ordinals:
  \quad
  \BlockKG{\rho}{n} > \BlockKG{\rho}{n+1}
  \,,
 \\\label{eq_mainK_cond_lim}
  \forall\rho\in[0,\SingK]:
  \quad
  \lim_{n\to\infty}\CondBlockKG{\rho}{n} = 1-\InvK{\rho}
  \,,
 \\\label{eq_mainK_free_con}
  \forall\rho\in[0,\SingK]:
  \quad
  \FreeK{\rho}= \log(1-\InvK{\rho})\,,
 \\\label{eq_mainK_free_div}
  \lim_{\rho\to{\SingK}^{-}}
  \frac{\partial}{\partial\rho} \FreeK{\rho}
  = \infty\text{ if }k\not=0\,.
\end{gather}
\end{subequations}
\end{Thm}

Statements~\eqref{eq_mainK_sing},~\eqref{eq_mainK_roots_mon} and~\eqref{eq_mainK_roots_lim} are already in~\cite{Liggett_Schonmann_Stacey__DominationByProductMeasures__AoP_1997}. Statement~\eqref{eq_mainK_cond_mon} follows from the general statement in~\cite{Temmel__ShearersMeasureAndStochasticDominationOfProductMeasures__JTP_2014}, with~\eqref{eq_mainK_block_mon} being a corollary. Statement~\eqref{eq_mainK_cond_lim} is from~\cite{Mathieu_Temmel__KIndependentPercolationOnTrees__SPA}. Statement~\eqref{eq_mainK_free_con} follows from~\eqref{eq_mainK_cond_lim} by telescoping and the monotonicity~\eqref{eq_mainK_cond_mon}. The proof of statement~\eqref{eq_mainK_free_div} is in proposition~\ref{prop_free_div_k}.

\subsection{The continuous case}
\label{sec_C}

The continuous case is on $(\Reals,\Modulus{.})$ with hard-core radius $1$ and integration with respect to the Lebesgue measure. The subscript $c$ marks related quantities.\\

Consider the map
\begin{equation}\label{eq_funC}
 \ContractionC:
  \quad[0,\infty[\to\left[0,\ValC\right]
  \quad y\mapsto y e^{-y}\,,
\end{equation}
For $\rho\in[0,\ValC]$, let $\InvC{\rho}$ be the unique pre-image under $\ContractionC$ in $[0,1]$.

\begin{Thm}\label{thm_shearerC}
Let $\rho\in[0,\ValC]$ and $\xi$ be a homogeneous Poisson PP with intensity $\InvC{\rho}$. Define a PP $\eta$ by the density
\begin{equation}\label{eq_construction_continuous}
 \eta(dx)=1 \Iff \xi(dx)=1\text{ and }\xi(]x-1,x[)=0\,.
\end{equation}
Then $\eta$ is Shearer's PP on $\Reals$ with intensity $\rho$ and distance $1$. Shearer's PP does not exist for larger intensities than $\ValC$.
\end{Thm}

The construction of Shearer's PP may be seen as a Matérn type point process~\cite{Matern__SpatialVariation_StochasticModelsAndTheirApplicationToSomeProblemsInForestSurveysAndOtherSamplingInvestigations__MFSS_1960} with a one-sided  simultaneous thinning rule: a point of the Poisson point process $\xi$ survives, if it has an empty open unit interval to its left. In analogy with the discrete case, we may call $\eta$ a \emph{$1$-block factor} of the PP $\xi$.\\

We look at $\BlockCG{\rho}{t}:=\GenFunAt{-\rho}{[0,t]}$, with smallest root $\RootC{r}$, and $\CondBlockCG{\rho}{s}{t}:=\frac{\BlockCG{\rho}{s+t}}{\BlockCG{\rho}{t}}$, if well defined. Let $\displaystyle \FreeC{\rho}:= \lim_{t\to\infty}\frac{\log\BlockCG{\rho}{t}}{t}$. If $\rho$ is clear from context, we omit it. The singularity reduces to
\begin{equation*}
 \SingC:=
 \sup\Set{\rho\ge 0\mid
  \forall s,t\in[0,\infty[: \CondBlockCG{\rho}{s}{t}>0
 }
 =\inf\Set{\RootC{t}\mid t\in[0,\infty[}\,.
\end{equation*}

\begin{Thm}\label{thm_mainC}
\begin{subequations}\label{eq_mainC}
\begin{gather}
 \label{eq_mainC_sing}
  \SingC = \ValC
  \,,
 \\\label{eq_mainC_roots_mon}
  \forall t\in[0,\infty[, s> 0:
  \quad
  \RootC{t}>\RootC{t+s}
  \,,
 \\\label{eq_mainC_roots_lim}
  \lim_{t\to\infty} \RootC{t} = \SingC
  \,,
 \\\label{eq_mainC_cond_mon}
  \begin{gathered}
  \forall \rho\in]0,\SingC], s,s',t,t'\in[0,\infty[\text{ with }s'+t'>0:
  \\\hspace{15em}
  \CondBlockCG{\rho}{s}{t}>\CondBlockCG{\rho}{s+s'}{t+t'}
  \,,
  \end{gathered}
 \\\label{eq_mainC_block_mon}
  \forall \rho\in]0,\SingC],t\in[0,\infty[, s\in]0,\infty[:
  \quad
  \BlockCG{\rho}{t}>\BlockCG{\rho}{t+s}
  \,,
 \\\label{eq_mainC_cond_bound}
  \forall\rho\in[0,\SingC],\forall t,s\in[0,\infty[:
  \quad
  \CondBlockCG{\rho}{s}{t}\ge\exp(-s\InvC{\rho})
  \,,
 \\\label{eq_mainC_cond_lim}
  \forall\rho\in[0,\SingC]:
  \quad
  \lim_{\substack{t\to\infty\\s\to 0}}
  \frac{\log\CondBlockCG{\rho}{s}{t}}{s} = -\InvC{\rho}
  \,,
 \\\label{eq_mainC_free_con}
  \forall\rho\in[0,\SingC]:\quad
  \FreeC{\rho} =
  \rho\left(1 + 2\rho\InvC{\rho} + \frac{\rho^2\InvC{\rho}^2}{2}\right)
  \,,
 \\\label{eq_mainC_free_div}
  \lim_{\rho\to{\SingC}^{-}} \frac{\partial}{\partial\rho} \FreeC{\rho}
  = \infty\,.
\end{gather}
\end{subequations}
\end{Thm}

Statements~\eqref{eq_mainC_roots_mon} and~\eqref{eq_mainC_cond_mon} are from~\cite{Temmel__ShearersPointProcessTheHardSphereGasAndAContinuumLovaszLocalLemma}. Statement~\eqref{eq_mainC_free_con} is from proposition~\ref{prop_free_cont_con}. Statement~\ref{eq_mainC_free_div} is from proposition~\ref{prop_free_cont_div}. Statement~\ref{eq_mainC_cond_lim} is from proposition~\ref{prop_cblock_lim}. Statement~\ref{eq_mainC_cond_bound} is from proposition~\ref{prop_lll_cont}.  Statement~\eqref{eq_mainC_sing} follows either from propositions~\ref{prop_lll_cont} and~\ref{prop_series_cont_div} or from proposition~\ref{prop_free_cont_con} and~\ref{prop_series_cont_div}.

\subsection{Scaling}
\label{sec_scaling}

This section discusses how the parts of the results for continuous case may be seen as scaled versions of the corresponding discrete results.\\

The singularities scale:
\begin{equation*}
 (k+1)\SingK
 = \left(1-\frac{1}{k+1}\right)^k
 \xrightarrow[k\to\infty]{}
 \ValC=\SingC\,.
\end{equation*}
Even more, the activities and their corresponding inverses, too. If $y\in[0,1]$, then $\frac{y}{k+1}\in[0,\frac{1}{k+1}]$ and
\begin{equation*}
 (k+1)\ContractionK\left(\frac{y}{k+1}\right)
 = y \left(1-\frac{y}{k+1}\right)^k
 \xrightarrow[k\to\infty]{} y e^{-y}
 = \ContractionC(y)\,.
\end{equation*}
Hence, for $\rho\in[0,\ValC]$,
\begin{equation*}
 (k+1)\InvK{\frac{\rho}{k+1}}\xrightarrow[k\to\infty]{} \InvC{\rho}\,.
\end{equation*}

For Shearer's point process, the scaling is even more obvious: The underlying Bernoulli product field scales to the Poisson point process of correct intensity $\InvC{\rho}$. The he dependent thinning rule is the same and the length of the free interval to the left scales with $\frac{1}{k+1}$.
Therefore, the critical values scale.
Note: both $(1+x)^{-k}$ and $(1-x)^k$ have the same scaling limit $ye^{-y}$ under $x=\frac{y}{k+1}$.\\

While in the continuous case there is only one fixed point equation, in the discrete case there are two: $\rho=x(1-x)^{k}$ is in~\eqref{eq_funK}, whereas $\rho=x(1-x)^{-(k+1)}=\mathfrak{h}_\CaseK(x)$ stems from the discrete one-sided tree-operator in section~\ref{sec_treeop_series}. As both $(1+x)^{-k-1}$ and $(1-x)^k$ have the same scaling limit $e^{-y}$ under $x=\frac{y}{k+1}$, these fixed point equations merge in the continuous setup to $\rho=ye^{-y}$~\eqref{eq_funC}.

\subsection{Previous results}
\label{sec_previous}

This section lists previous results obtained via cluster expansion. Previous results in the discrete case are
\begin{itemize}
 \item Dobrushin~\cite{Dobrushin__PerturbationMethodsOfTheTheoryOfGibbsianFields__LNM_1996}: $\SingK\ge \frac{(2k)^{2k}}{(2k+1)^{2k+1}}\sim\frac{1}{2ek}$
 \item special-casing Fernandéz and Procacci~\cite{Fernandez_Procacci__ClusterExpansionForAbstractPolymerModels_NewBoundsFromAnOldApproach__CMP_2007}: $\SingK\ge\frac{1}{\sqrt{2k(k+1)}+2k+1}\sim\frac{1}{(\sqrt{2}+2)k}$ (maximum of $\frac{\mu}{1+(2k+1)\mu+\binom{k+1}{2}\mu^2}$).
\end{itemize}
Previous results in the continuous case are
\begin{itemize}
 \item Ruelle~\cite{Ruelle__StatisticalMechanics_RigorousResults__Ben_1969} $\SingC\ge 1/(2e)$
 \item special-casing Fernandéz, Procacci and Scoppola~\cite{Fernandez_Procacci_Scoppola__TheAnalyticityRegionOfTheHardSphereGas_ImprovedBounds__JSP_2007}: $\SingC\ge \frac{1}{2+\sqrt{2}}$ (maximum of $\frac{\mu}{1+2\mu+\frac{\mu^2}{2}}$).
\end{itemize}

\section{The inductive approach}
\label{sec_ind}

\subsection{The discrete case}
\label{sec_ind_k}

Let $\rho\le\ValK$. The identity~\eqref{eq_fe} rewrites as
\begin{equation*}
 \CondBlockKR{n} = 1 - \frac{\rho}{\prod_{i=1}^k \CondBlockKG{\rho}{n-i}}\,.
\end{equation*}
We deduce~\eqref{eq_mainK_cond_mon} and show that
\begin{equation}
 \CondBlockKG{\rho}{\infty} := \lim_{n\to\infty} \CondBlockKG{\rho}{n}
\end{equation}
exists and fulfils the identity
\begin{equation}\label{eq_fe_cond_limit_k}
 \CondBlockKG{\rho}{\infty} = 1 - \frac{\rho}{\CondBlockKG{\rho}{\infty}^k}\,.
\end{equation}
This is just another form of the fixed-point of~\eqref{eq_funK}, whence $1-\CondBlockKG{\rho}{\infty}=\InvK{\rho}$.

\begin{Prop}\label{prop_free_div_k}
If $k\ge 1$, then
\begin{equation}
 \lim_{\rho\to{\SingK}^-} \frac{\partial}{\partial\rho}
 \lim_{n\to\infty} \CondBlockKG{\rho}{n} = \infty\,.
\end{equation}
For $k=0$, nothing happens.
\end{Prop}

\begin{proof}
We derive with respect to $\rho$ in~\eqref{eq_fe_cond_limit_k}, apply it to the result and obtain
\begin{equation*}
 \frac{\partial}{\partial\rho} \CondBlockKG{\rho}{\infty}
 = \frac{1}{\frac{k\rho}{\CondBlockKG{\rho}{\infty}}-\CondBlockKG{\rho}{\infty}^k}\,.
\end{equation*}
Because of~\eqref{eq_mainK_cond_lim} and $\CondBlockKG{\SingK}{\infty}=1-\InvK{\CondBlockKG{\rho}{\infty}}=\frac{k}{k+1}$, we have
\begin{equation*}
 \lim_{\rho\to\SingK^-}
 \frac{k\rho}{\CondBlockKG{\rho}{\infty}}-\CondBlockKG{\rho}{\infty}^k
 = \frac{k\ValK}{\frac{k}{k+1}} - \left(\frac{k}{k+1}\right)^k = 0\,.
\end{equation*}
\end{proof}

\subsection{The continuous case}
\label{sec_ind_c}

The following proposition is a special case of the continuous Lovász Local Lemma in~\cite{Temmel__ShearersPointProcessTheHardSphereGasAndAContinuumLovaszLocalLemma}. The identity~\eqref{eq_fe} has the form: if $s\in[0,1[$ and $t\in[0,\infty[$, then
\begin{equation}
 \BlockCG{\rho}{t+s} = \BlockCG{\rho}{t}
 - \rho\int_0^s \BlockCG{\rho}{(t+x-1)\MaxOp 0} dx\,.
\end{equation}
If $\BlockCG{\rho}{t}>0$, this rewrites into
\begin{equation}\label{eq_fe_cond_c}
 \CondBlockCG{\rho}{s}{t}
 = 1 - \rho\int_0^s \CondBlockCG{\rho}{(1-x)\MinOp t}{(t+x-1)\MaxOp 0}\TMO dx\,.
\end{equation}
If
\begin{equation}
 \CondBlockCG{\rho}{s}{\infty} :=
 \lim_{t\to\infty} \CondBlockCG{\rho}{s}{t}
\end{equation}
exists, then we have
\begin{equation}
 \CondBlockCG{\rho}{s}{\infty}
 = 1 - \rho\int_0^s \CondBlockCG{\rho}{1-x}{\infty}\TMO dx\,.
\end{equation}

\begin{Prop}\label{prop_lll_cont}
If $\rho\le\ValC$, then
\begin{equation*}
 \forall s,t\in[0,\infty[:
 \quad
 \CondBlockCG{\rho}{s}{t}\ge\exp(-s\InvC{\rho})\,.
\end{equation*}
\end{Prop}

\begin{proof}
We omit the trivial case with $\rho=0$. Fix $\rho>0$ and $\lambda:=\InvC{\rho}>0$. Let $\NuX{x}:=e^{x\lambda}$. This means that $\rho=\lambda\NuX{-1}$ and that $d\NuX{x}=\lambda \NuX{x}dx$. We use
\begin{equation*}
 \rho\int_a^b\NuX{1-x} dx = \NuX{-a}-\NuX{-b}\,.
\end{equation*}
We telescope to reduce to the case $s\in[0,1[$. We proceed by induction over $\Floor{s+t}$. We have three base cases and one induction step.\\

Case $\Floor{s+t}=0$: Consider $f_1(s,t):= 1- (s+t)\rho - (1-t\rho)\NuX{-s}$ on $[0,\infty[^2$. We have $f_1(0,0)=0$ and
\begin{equation*}
 \DiffVecOp{}f_1(s,t)
 =\begin{pmatrix}
   -\rho + (1-t\rho)\NuX{-s}
   \\-\rho + \rho\NuX{-s}
  \end{pmatrix}
 \ge\begin{pmatrix}
   -\ValC + (1-\ValC)
   \\\rho(\NuX{-s}-1)
  \end{pmatrix}
 \ge\begin{pmatrix}
   0
   \\0
  \end{pmatrix}
 \,.
\end{equation*}
Therefore, $f_1$ is non-negative on $[0,\infty[^2$ and
\begin{equation*}
 \CondBlockCR{s}{t}
 = \frac{\BlockCR{s+t}}{\BlockCR{t}}
 = \frac{1 - (s+t)\rho}{1-t\rho} \ge\NuX{-s}\,.
\end{equation*}

Case $\Floor{s+t}=1, t<1$: On $[0,\infty[^2$, consider the function $f_2(s,t):= 1- (s+t)\rho + \frac{1}{2}(s+t-1)^2\rho^2 - (1-t\rho)\NuX{-s} = f_1(s,t) + \frac{1}{2}(s+t-1)^2\rho^2\ge f_1(s,t)$. Therefore, $f_2$ is non-negative on $[0,\infty[^2$ and
\begin{equation*}
 \CondBlockCR{s}{t}
 = \frac{\BlockCR{s+t}}{\BlockCR{t}}
 = \frac{1 - (s+t)\rho + \frac{1}{2}(s+t-1)^2\rho^2}{1-t\rho} \ge\NuX{-s}\,.
\end{equation*}

Case $\Floor{s+t}=1, t>1$: For $x\in[0,s[$, we have $t+x-1\le t+s-1<1$. Use~\eqref{eq_fe_cond_c} and the previous case to get
\begin{equation*}
 \CondBlockCR{s}{t}
 ={} 1 - \rho\int_0^s \CondBlockCR{1-x}{t+x-1}\TMO dx
 \ge{} 1 - \rho \int_0^s \NuX{1-x} dx
 = \NuX{-s}\,.
\end{equation*}

Case $\Floor{s+t}\ge 2$: Use~\eqref{eq_fe_cond_c} to get
\begin{equation*}
 \CondBlockCR{s}{t}
 = 1 - \rho\int_0^s \CondBlockCR{1-x}{t+x-1}\TMO dx\,.
\end{equation*}
For each $x\in[0,s]$, we have $\Floor{s+t}>\Floor{t+s-1}\ge\Floor{t+x-1}$. We use~\eqref{eq_fe_cond_c} and the induction hypothesis to obtain
\begin{align*}
 \CondBlockCR{1-x}{t+x-1}
 &={} 1 - \rho\int_0^{1-x} \CondBlockCR{1-y}{t+x+y-2}\TMO dy
 \\&\ge{} 1 - \rho \int_0^{1-x}\NuX{1-y} dy = \NuX{-(1-x)}\,.
\end{align*}
We plug this into the first expansion and get
\begin{align*}
 \CondBlockCR{1-x}{t+x-1}
 &={} 1 - \rho\int_0^{1-x} \CondBlockCR{1-y}{t+x+y-2}\TMO dy
 \\&\ge{} 1 - \rho \int_0^s\NuX{1-x} dx = \NuX{-s}\,.
\end{align*}

\end{proof}

\section{Cluster expansion}
\label{sec_ce}

\subsection{Ursell coefficients}

This section recapitulates how we may express the Ursell coefficients $\Ursell{x_1,\dotsc,x_n}$ as an explicit sum of trees.\\

For a finite graph $G$, let $\SpanningGraphsOf{G}$ be the spanning subgraphs of $G$ and $\SpanningTreesOf{G}$ be the spanning trees of $G$. A \emph{partition scheme} $\PartitionScheme$ is a function $\SpanningTreesOf{G}\to\SpanningGraphsOf{G}$, such that the intervals $\Set{[T,\PartitionScheme(T)]: T\in\SpanningTreesOf{G}}$ partition the poset $(\SpanningGraphsOf{G},\subseteq)$ into disjoint and bounded lattices. The \emph{singleton trees} of $\PartitionScheme$ are its fixed points: $\SingletonTrees{G}:=\Set{T\in\SpanningTreesOf{G}:\PartitionScheme(T)=T}$\\

For $x_1,\dotsc,x_n\in\Space$, the \emph{cluster} $\Cluster{x_1,\dotsc,x_n}$ is the graph with vertices $\IntegerSet{n}$ and edges between vertices $i$ and $j$ with $\Metric(x_1,x_j)<R$. If we have a partition scheme of a cluster $\Cluster{x_1,\dotsc,x_n}$, then Penrose's theorem~\cite{Penrose__ConvergenceOfFugacityExpansionsForClassicalSystems__SMFA_1967} implies that
\begin{equation*}
 U(x_1,\dotsc,x_n)
 = \sum_{H\text{ spans }G(x_1,\dotsc,x_n)} (-1)^{\Cardinality{E(H)}}
 = (-1)^{n-1}\Cardinality{\SingletonTrees{\Cluster{x_1,\dotsc,x_n}}}\,.
\end{equation*}
Hence,
\begin{equation*}
  \log\GenFunAt{-\rho}{\FinVol}
 = -\sum_{n=1}^\infty \frac{\rho^n}{n!}
   \int_{\FinVol^n}
   \Cardinality{\SingletonTrees{\Cluster{x_1,\dotsc,x_n}}}
   \prod_{i=1}^n dx_i\,.
\end{equation*}

\subsection{The one-sided partition scheme}

This section describes the \emph{one-sided partition scheme}, a partition scheme adapted to a one-dimensional space. Proposition~\ref{prop_one_singleton_trees} describes its very nice and explicit set of singleton trees.\\

Fix $x_1,\dotsc,x_n\in\Space$. Let $G:=\Cluster{x_1,\dotsc,x_n}$. We root $G$ at the vertex $1$. Let $\SpanningTreesOf{n}$ be the set of trees with $n$ vertices. For $T\in\SpanningTreesOf{n}$ and $\vec{x}\in\FinVol^n$, we have
\begin{equation*}
 \Iverson{T\in\SingletonTrees{\Cluster{x_1,\dotsc,x_n}}}
 =\Iverson{\PartitionScheme(T)=T,T\subseteq\Cluster{x_1,\dotsc,x_n}}\,.
\end{equation*}
If $\Space$ is one-dimensional, i.e., it has a total order compatible with the distance, then there is a partition scheme with an explicit description of its singleton trees. For $T\in\SpanningTreesOf{G}$, let $l:\IntegerSet{n}\to\Cardinals$ denote the level of the vertices in $T$, i.e., the distance to the root $1$. If $l_i\ge 1$, let $p_i$ be the parent of $i$. If $l_i\ge 2$, let $g_i$ be the grandparent of $i$.

\begin{Prop}\label{prop_one_singleton_trees}
There exists a partition scheme $\SchemeOne$, such that
\begin{equation*}
 \Iverson{T\in\SingletonTrees{G}}
 = \prod_{i=1}^n
  \Bigl(
   \Iverson{l_i=1}
   + \Iverson{l_i>1,x_i\le x_{p_i}\le x_{g_i}}
   + \Iverson{l_i>1,x_i\ge x_{p_i}\ge x_{g_i}}
  \Bigr)
\end{equation*}
If we regard $T$ in $\SpanningTreesOf{n}$ instead of $\SpanningTreesOf{G}$, then we have to add the cluster constraints from $G$:
\begin{multline}\label{eq_one_singleton_trees}
 \Iverson{T\in\SingletonTrees{G}}
 = \prod_{i=2}^n \Iverson{\Modulus{x_i-x_{p_i}}<R}
 \\
  \times\prod_{i=2}^n\Bigl(
   \Iverson{l_T(i)=1}
   + \Iverson{l_i>1,x_i\le x_{p_i}\le x_{g_i}}
   + \Iverson{l_i>1,x_i\ge x_{p_i}\ge x_{g_i}}
  \Bigr)
\end{multline}
\end{Prop}

We call the partition scheme $\SchemeOne$ the one-sided scheme, because after the first-level children have chosen a direction (increasing or decreasing), the subtree based on such a child branches only in the given direction, with no additional constraint on and between further descendants.\\

The one-sided scheme $\SchemeOne$ is an explorative partition scheme~\cite[Section 5.2]{Temmel__SufficientConditionsForUniformBoundsInAbstractPolymerSystemsAndExplorativePartitionSchemes__JSP_2014}. An explorative partition scheme has two key ingredients: an exploration algorithm (selecting a tree from a cluster) and a tree edge complement partition (getting $\SchemeOne(T)$ from a tree $T$). We state these two ingredients below and omit the proof of their correctness and of proposition~\ref{prop_one_singleton_trees}, as the one-sided partition scheme is a variation of the returning scheme~\cite[Section 5.4]{Temmel__SufficientConditionsForUniformBoundsInAbstractPolymerSystemsAndExplorativePartitionSchemes__JSP_2014} and the proofs would be slight rewritings of those of~\cite[Propositions 24 \& 25]{Temmel__SufficientConditionsForUniformBoundsInAbstractPolymerSystemsAndExplorativePartitionSchemes__JSP_2014}. The key idea of the returning scheme is ``The exploration algorithm should select those edges, which we want the singleton trees to contain.''. For the one-sided scheme, this becomes: Once a direction is chosen, the exploration algorithm below tries to go in this direction as far as possible. See the $\Progressing$ case in algorithm~\ref{alg_exploration}. If you ever have to turn back (i.e., change direction), then the tree edge complement partition ensures with~\eqref{eq_conf_dir} that $\SchemeOne(T)\not=T$. The first level gets a much simpler exceptional treatment based on the greedy scheme~\cite[Section 5.3]{Temmel__SufficientConditionsForUniformBoundsInAbstractPolymerSystemsAndExplorativePartitionSchemes__JSP_2014}. For the remainder of this section, the notation follows~\cite[Section 5]{Temmel__SufficientConditionsForUniformBoundsInAbstractPolymerSystemsAndExplorativePartitionSchemes__JSP_2014}.

\begin{Alg}[Exploration algorithm]
\label{alg_exploration}
Let $H\in\SpanningGraphsOf{G}$. For every $k$, let $H_k$, $T_k$, $U_k$, $B_k$ and $P_k$ be as in \cite[generic exploration algorithm]{Temmel__SufficientConditionsForUniformBoundsInAbstractPolymerSystemsAndExplorativePartitionSchemes__JSP_2014}. The missing parts to construct $H_{k+1}$ from $H_k$ are:\\

Call an edge $(i,j)\in\Edges{C\cap P_k,B_k}\cap\Edges{H}$ \emph{progressing} (short $\Progressing$), if $k=0$, or $k\ge 1$ and either $x_{p_i}>x_i>x_j$ or $x_{p_i}<x_i<x_j$. All other edges in $\Edges{C\cap P_k,B_k}\cap\Edges{H}$ are \emph{reversing} (short $\Reversing$). A vertex is $i\in P_k$ progressing, if all edges $(i,j)$ are progressing and reversing, if there exists a reversing edge $(i,j)$. Finally we say that a connected component $C$ of $H_k|_{U_k}$ is progressing, if all vertices in $C\cap P_k$ are progressing and reversing, if $C\cap P_k$ contains at least one reversing vertex.\\

For $k=0$, we do the same as in the greedy partition scheme~\cite{Temmel__SufficientConditionsForUniformBoundsInAbstractPolymerSystemsAndExplorativePartitionSchemes__JSP_2014}. For $k\ge 1$, we do as follows:\\

If $C$ is an $\Progressing$ connected component of $H_k|_{U_k}$:
\begin{description}
 \item[\StepDirProgBoundary{}] \ActionSelect{} $C\cap S_k:= C\cap P_k$.
 \item[\StepDirProgIgnored{}] As $C\cap I_k=\emptyset$ \ActionRemove{} nothing.
 \item[\StepDirProgParent{}] For each $i\in C\cap S_k$, let $j_i:=\argmin\Set{j\in B_k: (i,j)\in \Edges{H_k}}$. \ActionSelect{} $(i,j_i)$.
 \item[\StepDirProgUncles{}] For each $i\in C\cap S_k$, \ActionRemove{} all $(i,j)\in\Edges{H_k}$ with $j_i\not=j\in B_k$.
 \item[\StepDirProgCousins{}] \ActionRemove{} all of $\Edges{C\cap S_k}\cap\Edges{H_k}$.
\end{description}

If $C$ is a $\Reversing$ connected component of $H_k|_{U_k}$:
\begin{description}
 \item[\StepDirRevBoundary{}] \ActionSelect{} $C\cap S_k:= \Set{i\in C\cap P_k: i\text{ is }\Reversing}$.
 \item[\StepDirRevIgnored{}] $C\cap I_k=\Set{i\in C\cap P_k: i\text{ is }\Progressing}$. \ActionRemove{} all of $\Edges{B_k,(C\cap I_k)}\cap\Edges{H_k}$.
 \item[\StepDirRevParent{}] For each $i\in C\cap S_k$, let
 $j_i:=\argmin\Set{j\in B_k: (i,j)\in \Edges{H_k}\text{ is }\Reversing}$.
  \ActionSelect{} $(i,j_i)$.
 \item[\StepDirRevUnclesReversing{}] For each $i\in C\cap S_k$, \ActionRemove{} every $\Reversing$ $(i,j)\in\Edges{H_k}$ with $j_i\not=j\in B_k$.
 \item[\StepDirRevUnclesReversing{}] For each $i\in C\cap S_k$, \ActionRemove{} every $\Progressing$ $(i,j)\in\Edges{C\cap S_k,B_k}\cap\Edges{H_k}$.
 \item[\StepDirRevCousins{}] \ActionRemove{} all of $\Edges{C\cap S_k}\cap\Edges{H_k}$.
\end{description}
\end{Alg}

\begin{Alg}[Tree edge complement partition]
\label{alg_complement}
Let $\Tree\in\SpanningTrees{G}$. Let $L_k$ be the $k^{th}$ level of $\Tree$. First we determine if an edge $(\Parent{i},i)$ is a progressing (short $\Progressing$) or reversing (short $\Reversing$) edge:
\begin{equation}\label{eq_dir_status}
 \SideStatus:
 \quad I\setminus\Set{o}\setminus L_1\to\Set{\Progressing,\Reversing}
 \quad i\mapsto\begin{cases}
   \Progressing
   &\text{if }x_i\le x_{\Parent{i}}\le x_{\Parent{\Parent{i}}}
    \text{ or }x_i\ge x_{\Parent{i}}\ge x_{\Parent{\Parent{i}}}\\
   \Reversing
   &\text{else.}\,.
  \end{cases}
\end{equation}
For $k\ge 2$ define the \emph{equivalence relation} $\SideEquivalentK$ on $L_k$ by
\begin{equation}\label{eq_eqrel_dir}
 i\SideEquivalentK{}j
 \Iff
 \SideStatus(\TreePath{o}{i}\setminus\Set{o})
  =\SideStatus(\TreePath{o}{j}\setminus\Set{o})\,,
\end{equation}
where the equality on the rhs is taken in $\Set{\Progressing,\Reversing}^{k-1}$ between the labels of the paths $\TreePath{o}{.}$ to the root. This implies that an equivalence class consists of either only same or only non-same nodes and whence we can extend $\SideStatus$ to them. For completeness let $\SideEquivalent{0}$ be the trivial equivalence relation on $L_0$ and be totally disjoint on $L_1$. The equivalence classes possess a \emph{tree structure} consistent with $\Tree$:
\begin{equation}\label{eq_eqrel_dir_class_tree}
 i\SideEquivalentKPlusOne j
 \Then
 \Parent{i}\SideEquivalentK\Parent{j}\,,
\end{equation}
that is equivalent vertices in $L_{k+1}$ have equivalent parents in $L_k$. We therefore call $\SideClass{\Parent{i}}{k}$ the \emph{parent class} of $\SideClass{i}{k+1}$.\\

We partition $E\setminus\Edges{\Tree}$ into $\AdmissibleEdgesTree{\SchemeOne}\uplus\ConflictingEdgesTree{\SchemeOne}$. Edges incident to the root $o$ are treated as in the greedy scheme in~\cite{Temmel__SufficientConditionsForUniformBoundsInAbstractPolymerSystemsAndExplorativePartitionSchemes__JSP_2014}.\\

Let $1\le k\le l$, $j\in L_k$, $i\in L_l$ and $e:=(i,j)\in E\setminus\Edges{\Tree}$. Then $e\in\ConflictingEdgesTree{\SchemeOne}$, iff one of the mutually exclusive conditions~\eqref{eq_conf_dir} holds:
\begin{subequations}\label{eq_conf_dir}
\begin{align}
 &\label{eq_conf_dir_notClassPath}
 \qquad\qquad\qquad
 \SideClass{j}{k}\not\in\TreePath{\SideClass{o}{0}}{\SideClass{i}{l}}
 \,,\\
 &\label{eq_conf_dir_classAncestorReversing}
 l\ge 2
  \,\land\,
 \SideClass{j}{k}\in\TreePath{\SideClass{o}{0}}{\SideClass{\Parent{\Parent{i}}}{l-2}}
  \,\land\,
 x_i\not=x_j
 \,,\\
 &\label{eq_conf_dir_classAncestorProgressing}
 l\ge 2
  \,\land\,
 \SideClass{j}{k}\in\TreePath{\SideClass{o}{0}}{\SideClass{\Parent{\Parent{i}}}{l-2}}
  \,\land\,
 x_i=x_j
  \,\land\,
 \SideStatus(C)=\Progressing
 \,,\\
\intertext{
 where $C\in\TreePath{\SideClass{j}{k}}{\SideClass{\Parent{i}}{l-1}}$ the unique class with $\Parent{C}=\SideClass{j}{k}$,
}
 &\label{eq_conf_dir_smallUncleReversing}
 l\ge 1
  \,\land\,
 \SideClass{j}{k}=\SideClass{\Parent{i}}{l-1}
  \,\land\,
 x_i\not=x_j
  \,\land\,
 \SideStatus(i)=\Reversing
  \,\land\,
 j<\Parent{i}
 \,,\\
 &\label{eq_conf_dir_differentUncleProgressing}
 l\ge 1
  \,\land\,
 \SideClass{j}{k}=\SideClass{\Parent{i}}{l-1}
  \,\land\,
 x_i\not=x_j
  \,\land\,
 \SideStatus(i)=\Progressing
 \,,\\
 &\label{eq_conf_dir_smallUncleProgressing}
 l\ge 1
  \,\land\,
 \SideClass{j}{k}=\SideClass{\Parent{i}}{l-1}
  \,\land\,
 x_i=x_j
  \,\land\,
 \SideStatus(i)=\Progressing
  \,\land\,
 j<\Parent{i}
 \,.
\end{align}
\end{subequations}
And $e\in\AdmissibleEdgesTree{\SchemeOne}$, iff one of the mutually exclusive conditions~\eqref{eq_adm_dir} holds:
\begin{subequations}\label{eq_adm_dir}
\begin{align}
 &\label{eq_adm_dir_equalClass}
 \qquad\qquad\qquad
 \SideClass{j}{k}=\SideClass{i}{l}
 \,,\\
 &\label{eq_adm_dir_classAncestor}
 l\ge 2
  \,\land\,
 \SideClass{j}{k}\in\TreePath{\SideClass{o}{0}}{\SideClass{\Parent{\Parent{i}}}{l-2}}
  \,\land\,
 x_i=x_j
  \,\land\,
 \SideStatus(C)=\Reversing
 \,,\\
\intertext{
 where $C\in\TreePath{\SideClass{j}{k}}{\SideClass{\Parent{i}}{l-1}}$ the unique class with $\Parent{C}=\SideClass{j}{k}$,
}
 &\label{eq_adm_dir_reversingUncleReversing}
 l\ge 1
  \,\land\,
 \SideClass{j}{k}=\Parent{\SideClass{i}{l}}
  \,\land\,
 x_i\not=x_j
  \,\land\,
 \SideStatus(i)=\Reversing
  \,\land\,
 j>\Parent{i}
 \,,\\
 &\label{eq_adm_dir_progessingUncleReversing}
 l\ge 1
  \,\land\,
 \SideClass{j}{k}=\Parent{\SideClass{i}{l}}
  \,\land\,
 x_i=x_j
  \,\land\,
 \SideStatus(i)=\Reversing
 \,,\\
 &\label{eq_adm_dir_uncleProgressing}
 l\ge 1
  \,\land\,
 \SideClass{j}{k}=\Parent{\SideClass{i}{l}}
  \,\land\,
 x_i=x_j
  \,\land\,
 \SideStatus(i)=\Progressing
  \,\land\,
 j>\Parent{i}
 \,.
\end{align}
\end{subequations}
\end{Alg}

\subsection{Tree-operator series}
\label{sec_treeop_series}

This section builds heavily on \cite[Section 4]{Temmel__SufficientConditionsForUniformBoundsInAbstractPolymerSystemsAndExplorativePartitionSchemes__JSP_2014}, in particular~\cite[Proposition 7]{Temmel__SufficientConditionsForUniformBoundsInAbstractPolymerSystemsAndExplorativePartitionSchemes__JSP_2014}. We encode the properties of the singleton trees of the one-sided partition scheme in proposition~\ref{prop_one_singleton_trees} into convergence conditions for cluster expansion series. In this section, we often give three expressions: a generic one without subscript, followed by the specialised ones for the discrete and continuous cases.\\

In the tree-operator of the Penrose partition scheme $\SchemePenrose$~\cite{Penrose__ConvergenceOfFugacityExpansionsForClassicalSystems__SMFA_1967}\cite[Sections 5.1 \& 5.3]{Temmel__SufficientConditionsForUniformBoundsInAbstractPolymerSystemsAndExplorativePartitionSchemes__JSP_2014} a vertex has $s$ children with weight $G(s)$:
\begin{equation*}
 G(s):= \sum_{n=0}^\infty \frac{\Iverson{s=n}}{n!}
  \int_{]-R,R[^n}\IsRHC{\vec{y}}d\vec{y}\,,
\end{equation*}
\begin{equation*}
 G_k(s):=\Iverson{s=0} + (2k+1)\Iverson{s=1} + \binom{k+1}{2}\Iverson{s=2}\,,
\end{equation*}
\begin{equation*}
 G_\CaseC(s) := \Iverson{s=0} + 2\Iverson{s=1} + \frac{\Iverson{s=2}}{2}\,,
\end{equation*}
leading to a tree-operator defined by
\begin{equation*}
 g(\mu):=\sum_{s=0}^\infty G(s) \mu^s\,,
\end{equation*}
\begin{equation*}
 g_k(\mu):= 1 + (2k+1)\mu + \binom{k+1}{2}\mu^2\,,
\end{equation*}
\begin{equation*}
 g_\CaseC(\mu) := 1 + 2\mu + \frac{\mu^2}{2}\,.
\end{equation*}
The function
\begin{equation*}
 \mathfrak{g}(\mu):=\frac{\mu}{g(\mu)}
 \qquad
 \mathfrak{g}_k(\mu):=\frac{\mu}{g_k(\mu)}
 \qquad
 \mathfrak{g}_\CaseC(\mu):=\frac{\mu}{g_\CaseC(\mu)}
\end{equation*}
has a global maximum in $[0,\infty[$ at
\begin{equation*}
 \mathfrak{g}_k\left(\sqrt{\frac{2}{k(k+1)}}\right)
 = \frac{\sqrt{k(k+1)}}{\sqrt{2k(k+1)}+2k+1}
 \qquad
 \mathfrak{g}_\CaseC(\sqrt{2}) = \frac{1}{2+\sqrt{2}}\,.
\end{equation*}

In the tree-operator of the one-sided partition scheme $\SchemeOne$ a non-root vertex has $s$ children with weight $H(s)$:
\begin{equation*}
 H(s):=
  \sum_{n=0}^\infty \frac{\Iverson{s=n}}{n!}
  \int_{[0,R[^n} d\vec{y}\,,
\end{equation*}
\begin{equation*}
 H_k(s) :=
  \sum_{n=0}^\infty \frac{\Iverson{s=n}}{n!}
  \sum_{i=1}^s \sum_{x_i\in\IntegerSet{k+1}\setminus\Set{x_1,\dotsc,x_{i-1}}} 1
 = \sum_{I\in\binom{\IntegerSet{k+1}}{s}} 1
 = \binom{k+1}{s}\,,
\end{equation*}
\begin{equation*}
 H_\CaseC(s):=
 \sum_{n=0}^\infty \frac{\Iverson{s=n}}{n!} \int_{[0,1[^n} d\vec{y}
 = \frac{R^s}{s!}\,,
\end{equation*}
leading to a tree-operator defined by
\begin{equation*}
 h(\mu):=\sum_{s=0}^\infty H(s) \mu^s\,,
\end{equation*}
\begin{equation*}
 h_k(\mu)
 :=\sum_{s=0}^\infty H_k(s)\mu^s
 = \sum_{s=0}^{k+1} \binom{k+1}{s} \mu^s
 = (1+\mu)^{k+1}\,,
\end{equation*}
\begin{equation*}
 h_\CaseC(\mu) := \sum_{s=0}^\infty H_\CaseC(s) \mu^s
 = \sum_{s=0}^\infty \frac{R^s \mu^s}{s!}
 = e^{R\mu}\,.
\end{equation*}
The function
\begin{equation*}
 \mathfrak{h}(\mu):=\frac{\mu}{h(\mu)}
 \qquad
 \mathfrak{h}_k(\mu):=\frac{\mu}{h_k(\mu)}
 \qquad
 \mathfrak{h}_\CaseC(\mu):=\frac{\mu}{h_\CaseC(\mu)}
\end{equation*}
has a global maximum in $[0,\infty[$ at
\begin{equation*}
 \mathfrak{h}_k\left(\frac{1}{k}\right) = \ValK
 \qquad
 \mathfrak{h}_\CaseC(1) = \ValC\,.
\end{equation*}

Consider the fix-point equation
\begin{equation*}
 \mu=\rho h(\mu)
 \qquad
 \mu=\rho h_k(\mu)
 \qquad
 \mu=\rho h_\CaseC(\mu)
 \,.
\end{equation*}
As the derivative of the rhs is positive for $\mu<\frac{1}{k}$ and $\mu<1$:
\begin{equation*}
 \frac{\partial\mu}{\partial\rho} = \frac{\mu}{1-k\mu}
 \qquad
 \frac{\partial\mu}{\partial\rho}=\frac{e^{\mu}}{1-\mu}\,.
\end{equation*}
Therefore, the inverse $\mu(\rho)$ is well-defined up to the maximum. In the continuous case, the fix-point equation is $\rho = \mu e^{-\mu}=\ContractionC(\mu)$, the derivative is$\frac{\partial\rho}{\partial\mu}=(1-\mu)e^{-\mu}$ and the inverse is $\mu_\CaseC(\rho)=\InvC{\rho}$.\\

Consider the operator $T_\rho:\mu\mapsto\rho h(\mu)$. By~\cite[Proposition 7]{Temmel__SufficientConditionsForUniformBoundsInAbstractPolymerSystemsAndExplorativePartitionSchemes__JSP_2014}, if $\rho\le\SingG$, then the generating function $Q$ of rooted, one-sided trees converges and equals
\begin{equation*}
 Q(\rho) := \lim_{n\to\infty} T_\rho^n(0)
 = \sum_{n=1}^\infty \frac{\rho^n}{n!}
  \sum_{T\in\SpanningTreesOf{n}} \prod_{i=1}^n H(s_i)
 = \rho\mu(\rho)\,.
\end{equation*}
For $n\in\Cardinals$, let
\begin{equation*}
 D(\rho,n):= \sum_{T\in\SpanningTreesOf{n}} G(s_1)\prod_{i=2}^n H(s_i)\,.
\end{equation*}
The full series $P$ of all singleton trees for the directed scheme needs and extra step for the first level (children of the root):
\begin{equation*}
 P(\rho) := \sum_{n=1}^\infty \frac{\rho^n}{n!}
  \sum_{T\in\SpanningTreesOf{n}} G(s_1)\prod_{i=2}^n H(s_i)
  = \sum_{n=1}^\infty \frac{\rho^n}{n!} D(\rho,n)\,.
\end{equation*}
We also have the truncated series
\begin{equation*}
 P_N(\rho):=\sum_{n=1}^N \frac{\rho^n}{n!} D(\rho,n)\,.
\end{equation*}

Due to measurability, the resulting structure differs slightly between the discrete and continuous cases. In the discrete case, we have
\begin{equation*}
 Q_\CaseK(\rho) = \rho\mu_\CaseK(\rho)
\end{equation*}
and the full series $P_\CaseK$ fulfils
\begin{equation*}
 P_\CaseK(\rho)
 = \rho\left(1+P_\CaseK(\rho)+2k Q_\CaseK(\rho)
   +\binom{k+1}{2} Q_\CaseK(\rho)^2\right)\,.
\end{equation*}
The rhs is not $g_\CaseK(Q_\CaseK(\rho))$, because a child of the root staying at the same integer than its parent has not chosen a side yet. Thus,
\begin{equation*}
 P_\CaseK(\rho) = \frac{\rho}{1-\rho}
  \left(1+2k Q_\CaseK(\rho)+\binom{k+1}{2} Q_\CaseK(\rho)^2\right)\,.
\end{equation*}
In the continuous case, we have
\begin{equation*}
 Q_\CaseC(\rho) = \rho\mu_\CaseC(\rho) = \rho\InvC{\rho}
\end{equation*}
and the full series $P_\CaseC$ is
\begin{equation*}
 P_\CaseC(\rho) = \rho g_\CaseC(Q_\CaseC(\rho))
 = \rho\left(1 + 2Q_\CaseC(\rho) + \frac{Q_\CaseC(\rho)^2}{2}\right)\,.
\end{equation*}

\begin{Prop}\label{prop_series_cont_div}
If $\rho>\ValC$, then $Q_\CaseC(\rho)$ and $P_\CaseC(\rho)$ diverge.
\end{Prop}

\begin{proof}
The trajectory of $0$ under the map $\mu\mapsto\rho e^{\mu}$ diverges.
\end{proof}

\subsection{The shape of the free energy}

\begin{Prop}\label{prop_free_cont_con}
For $\rho\in[0,\ValC]$, we have
\begin{equation}\label{eq_free_cont}
 \FreeC{\rho}
 :=\lim_{t\to\infty} -\frac{\log\BlockCG{\rho}{t}}{t}
 = P_\CaseC(\rho)\,.
\end{equation}
\end{Prop}

\begin{proof}
Fix $t\ge0$ and $\rho\in[0,\ValC]$. For $T\in\SpanningTreesOf{n}$, let
\begin{equation*}
 T(t):=
 \int_{[0,t]^n} \Iverson{T\in\SingletonTrees{\Cluster{\vec{x}}}} d\vec{x}
 \,.
\end{equation*}
We also have
\begin{equation*}
 R(t,n):=
 \int_{[0,t]^n} \Cardinality{\SingletonTrees{\Cluster{\vec{x}}}} d\vec{x}
 = \int_{[0,t]^n} \sum_{T\in\SpanningTreesOf{n}}
  \Iverson{T\in\SingletonTrees{\Cluster{\vec{x}}}} d\vec{x}
 = \sum_{T\in\SpanningTreesOf{n}} T(t)
\end{equation*}
and
\begin{equation*}
 S(\rho,t) := -\log\BlockCG{\rho}{t}
 = \sum_{n=1}^\infty \frac{\rho^n}{n!}
 \int_{[0,t]^n} \Cardinality{\SingletonTrees{\Cluster{\vec{x}}}} d\vec{x}
 = \sum_{n=1}^\infty \frac{\rho^n}{n!} R(t,n)
 \,.
\end{equation*}

Summing in from the leaves à la Cammarotta~\cite{Cammarota__DecayOfCorrelationsForInfiniteRangeInteractionsInUnboundedSpinSystems__CMP_1982} we bound
\begin{equation*}
 T(t)\le t\, G_\CaseC(s_1)\prod_{i=2}^n H_\CaseC(s_i)
 \quad\text{ and }\quad
 R(t,n)
 \le t D_\CaseC(\rho,n)
 \,.
\end{equation*}
This implies
\begin{equation*}
 S(\rho,t)
 \le t \sum_{n=1}^\infty \frac{\rho^n}{n!} D_\CaseC(\rho,n)
 = t P_\CaseC(\rho)\,.
\end{equation*}

On the other hand, if the depth of the tree $T$ is at most $d$ and the root label $x_1\in[d-1,r-d+1]$, then the vertex labels of $T$ are contained in $[0,t]$. Hence,
\begin{equation*}
 T(t) \ge (t-2(d-1))^+ G_\CaseC(s_1)\prod_{i=2}^n H_\CaseC(s_i)\,.
\end{equation*}
A size $n$ tree has at most depth $(n-1)$:
\begin{equation*}
 R(t,n) \ge (t-2(d-1))^+ \Iverson{n\le d} D_\CaseC(\rho,n)\,.
\end{equation*}
Hence,
\begin{align*}
 S(\rho,r)
 \ge{}& \sum_{n=1}^\infty \frac{\rho^n}{n!}
  (t-2(d-1))^+ \Iverson{n\le d} D_\CaseC(\rho,n)
 \\={}& (t-2(d-1))^+ \sum_{n=1}^d \frac{\rho^n}{n!} D_\CaseC(\rho,n)
 \\={}& (t-2(d-1))^+ P_{d,c}(\rho)\,.
\end{align*}
For small $\varepsilon>0$, we may choose first $d_\varepsilon$ such that $P_{d_\varepsilon,c}(\rho)\ge (1-\varepsilon)P_\CaseC(\rho)$ and then $t_\varepsilon$ such that $(t_\varepsilon-2(d_\varepsilon-1)R)^+\ge (1-\varepsilon)t_\varepsilon$. Thus, for all $d\ge d_\varepsilon$ and $t\ge t_\varepsilon$, we have
\begin{equation*}
 S(\rho,t)\ge (t-2(d-1))^+ P_{d,c}(\rho)
 \ge (1-\varepsilon)^2 t P_\CaseC(\rho)\,.
\end{equation*}
\end{proof}

\begin{Prop}\label{prop_cblock_lim}
For $\rho\in[0,\ValC]$, we have
\begin{equation*}
  \lim_{s\to 0, t\to\infty} -\frac{\log\CondBlockCG{\rho}{s}{t}}{s}
  = \InvC{\rho}\,.
\end{equation*}
\end{Prop}

\begin{proof}
Fix $\rho\in]0,\ValC]$, as the statement is trivial for $\rho=0$. We have
\begin{align*}
 -\log\CondBlockCR{s}{t}
 ={}&-\log\BlockCR{s+t}+\log\BlockCR{t}
 \\={}& \sum_{n=1}^\infty \frac{\rho^n}{n!}
  \int_{[0,s+t]^\setminus[0,s]^n}
   \Cardinality{\SingletonTrees{\Cluster{\vec{x}}}}
  d\vec{x}
 \\={}& \sum_{n=1}^\infty \frac{\rho^n}{n!}
  \int_0^s n
   \int_{[y,s+t]^{n-1}}
    \Cardinality{\SingletonTrees{\Cluster{y,\vec{x}}}}
   d\vec{x}
  dy
 \\={}& \rho \int_0^s
  \sum_{n=0}^\infty \frac{\rho^n}{n!}
   \int_{[y,s+t]^n}
    \Cardinality{\SingletonTrees{\Cluster{y,\vec{x}}}}
   d\vec{x}
  dy
\end{align*}
By~\cite[Proposition 7]{Temmel__SufficientConditionsForUniformBoundsInAbstractPolymerSystemsAndExplorativePartitionSchemes__JSP_2014}, we have that
\begin{equation*}
 \sum_{n=0}^\infty \frac{\rho^n}{n!}
  \int_{[0,\infty[^n}
    \Cardinality{\SingletonTrees{\Cluster{y,\vec{x}}}}
   d\vec{x}
 = \sum_{n=0}^\infty \frac{\rho^n}{n!}
  \sum_{T\in\SpanningTreesOf{n+1}} \prod_{i=1}^{n+1} H(s_i)
 = \frac{\InvC{\rho}}{\rho}\,.
\end{equation*}
On the one hand, we majorate as follows:
\begin{align*}
 -\log\CondBlockCR{s}{t}
 \le{}& \rho \int_0^s \sum_{n=0}^\infty \frac{\rho^n}{n!}
  \int_{[y,\infty[^n}
    \Cardinality{\SingletonTrees{\Cluster{y,\vec{x}}}}
   d\vec{x}
  dy
 \\={}& \rho s \sum_{n=0}^\infty \frac{\rho^n}{n!}
  \int_{[0,\infty[^n}
    \Cardinality{\SingletonTrees{\Cluster{y,\vec{x}}}}
   d\vec{x}
 \\={}&s\InvC{\rho}\,.
\end{align*}
On the other hand, we minorate as follows:
\begin{align*}
 -\log\CondBlockCR{s}{t}
 \ge{}& \rho \int_0^s \sum_{n=0}^{\Floor{t}} \frac{\rho^n}{n!}
  \int_{[y,\infty[^n}
    \Cardinality{\SingletonTrees{\Cluster{y,\vec{x}}}}
   d\vec{x}
  dy
 \\={}& \rho s \sum_{n=0}^{\Floor{t}} \frac{\rho^n}{n!}
  \int_{[y,\infty[^n}
    \Cardinality{\SingletonTrees{\Cluster{y,\vec{x}}}}
   d\vec{x}
  dy\,.
\end{align*}
Therefore, the limits
\begin{equation*}
 \lim_{s\to 0}\lim_{t\to\infty} -\frac{\log\CondBlockCR{s}{t}}{s}
 = \lim_{s\to 0} \frac{\InvC{\rho}}{s}
 = \InvC{\rho}
\end{equation*}
and
\begin{equation*}
 \lim_{t\to\infty}\lim_{s\to 0} -\frac{\log\CondBlockCR{s}{t}}{s}
 = \lim_{t\to 0} \rho \sum_{n=0}^{\Floor{t}} \frac{\rho^n}{n!}
  \int_{[y,\infty[^n}
    \Cardinality{\SingletonTrees{\Cluster{y,\vec{x}}}}
   d\vec{x}
  dy
 = \InvC{\rho}
\end{equation*}
are equal.
\end{proof}

\begin{Prop}\label{prop_free_cont_div}
We have
\begin{equation*}
 \lim_{\rho\to{\SingC}^{-}} \frac{\partial}{\partial\rho} \FreeC{\rho}
  = \infty\,.
\end{equation*}
\end{Prop}

\begin{proof}
We have
\begin{equation*}
 \frac{\partial}{\partial\rho}\InvC{\rho}
 = \frac{e^{\InvC{\rho}}}{1-\InvC{\rho}}\,.
\end{equation*}
This explodes at $\ValC$, as $\InvC{\ValC}=1$.
\end{proof}

\bibliographystyle{plain}
\bibliography{ref}

\end{document}